\documentclass[a4paper,onecolumn,10pt]{IEEEtran}
\pdfoutput=1
\usepackage[utf8]{inputenc}

\usepackage[english]{babel}
\usepackage[T1]{fontenc}
\usepackage{amsmath}
\usepackage[unicode=true,bookmarks=true,bookmarksnumbered=false,bookmarksopen=false,breaklinks=false,pdfborder={0 0 1}, backref=false,colorlinks=true]{hyperref}

\usepackage{xcolor}
\definecolor{myred}{rgb}{0.8,0.3,0.2}
\definecolor{myredfill}{rgb}{1,0.93,0.93}

\definecolor{myred2}{rgb}{0.6,0.5,0.3}
\definecolor{myredfill2}{rgb}{0.93,0.83,0.83}

\definecolor{myblue}{rgb}{0.05,0.25,0.7}
\definecolor{mybluefill}{rgb}{0.92,0.98,1}
\definecolor{mygreen}{rgb}{0.3,0.6,0.4}
\definecolor{mygreenfill}{rgb}{0.93,1,0.97}
\definecolor{mygrey}{rgb}{0.4,0.4,0.4}
\definecolor{mygreyfill}{rgb}{0.95,0.95,0.95}
\definecolor{mypurple}{rgb}{0.6,0.2,0.6}
\definecolor{mypurplefill}{rgb}{0.99,0.94,0.99}
\definecolor{myyellow}{rgb}{9,0.9,0.3}
\definecolor{myyellowfill}{rgb}{1,1,0.92}
\definecolor{myorange}{rgb}{9,0.7,0.3}
\definecolor{myorangefill}{rgb}{1,0.95,0.92}
\definecolor{mybrown}{rgb}{0.4,0.25,0.1}
\definecolor{mybrownfill}{rgb}{1,0.9,0.8}

\definecolor{mediumgrey}{rgb}{0.7,0.7,0.7}

\definecolor{linkred}{rgb}{0.6,0.1,0.1}
\definecolor{citeblue}{rgb}{0.2,0.35,0.75}
\definecolor{urlblue}{rgb}{0.2,0.25,0.45}

\hypersetup{
     linkcolor = linkred,
     anchorcolor = linkred,
     citecolor = citeblue,
     filecolor = linkred,
     urlcolor = urlblue
     }

\usepackage{bbm}
\usepackage{stmaryrd}
\usepackage{tabu,multirow}
\usepackage{float} 
\usepackage{lipsum}
\usepackage[autostyle]{csquotes}
\usepackage{tikz-cd}
\usetikzlibrary{decorations.pathreplacing,calligraphy}
\usepackage{amsmath,amsthm,amssymb,latexsym}
\usepackage{mathrsfs}
\usepackage{hhline}
\usepackage[caption=false]{subfig}
\usetikzlibrary{backgrounds}
\usepackage{graphicx}
\usepackage{authblk}
\usepackage{braket}
\usepackage{enumitem}
\setlist[enumerate]{label*=\arabic*.}
\usepackage[linesnumbered,ruled,vlined]{algorithm2e}

\usepackage[numbers,sort&compress]{natbib}

\newtheorem{theorem}{Theorem}

\newtheorem{lemma}{Lemma}
\newtheorem*{lemma*}{Lemma}

\newtheorem{proposition}{Proposition}

\begin{document}

\title{\Large{Monotones in Resource Theories for Dynamical Decoupling}}

\author[1]{Graeme D. Berk} 
\author[2,3]{Simon Milz}
\author[4]{Kavan Modi$^*$}

\affil[1]{Nanyang Quantum Hub, School of Physical and Mathematical Sciences, Nanyang Technological University, Singapore 639673.}
\affil[2]{School of Physics, Trinity College Dublin, Dublin 2, Ireland}
\affil[3]{Trinity Quantum Alliance, Unit 16, Trinity Technology and Enterprise Centre, Pearse Street, Dublin 2, D02YN67, Ireland}
\affil[4]{School of Physics and Astronomy, Monash University, Victoria 3800, Australia
\newline {\small *kavan@quantumlah.org}}

\date{\today}

\maketitle

\begin{abstract}
In Ref.~\cite{extractingdynamicalquantumresources} we presented a generalised dynamical resource theory framework that enabled noise reduction techniques including dynamical decoupling (DD) to be studied. While this fundamental contribution remains correct, it has been found that the main resource quantifiers we employed to study these resource theories -- based on the relative entropies between Choi states of multitime processes -- are not monotonic under the allowed transformations. In this letter we detail modified relative entropy-based resource quantifiers, prove that they are indeed monotonic in our resource theories. We re-interpret our numerical results in terms of these new relative entropy monotones, arriving at the same empirical conclusions: DD can be understood as temporal resource distillation, and improvements to noise reduction via our multitimescale optimal dynamical decoupling (MODD) method coincide with a decrease in the corresponding non-Markovianity monotone.
\end{abstract}

\section{Introduction}
Dynamical decoupling (DD) is an open-loop noise reduction method that uses a rapid sequence of unitary interventions to reduce the deleterious influence of a noisy environment on a quantum system~\cite{dynamicaldecouplingofopenquantumsystems}. It has long been understood that the efficacy of DD is linked to the non-Markovianity of the noise process~\cite{nonmarkoviannoisethatcannotbe,dynamicaldecouplingefficiencyversusquantumnonmarkovianity}, however a formal description of this relationship has been elusive. 

In Ref.~\cite{extractingdynamicalquantumresources}, we set out to describe DD in terms of quantum resource theories~\cite{review}. The starting point of Ref.~\cite{extractingdynamicalquantumresources} was a resource theory framework for multitime quantum processes~\cite{resourcetheoriesofmultitime}, using the process tensor formalism~\cite{nonmarkovianquantumprocesses}. In this framework, multitime processes can be represented, using the Choi isomorphism, as multipartite quantum states satisfying a hierarchy of causality trace conditions. Transformations in these resource theories are called superprocesses, and generally correspond to pre- and post-processing at each time. To understand DD as `resource distillation in time', we added a notion of temporal coarse-graining by which the number of temporal subsystems of the Choi state -- physically the number of times the agent is allowed to act at -- is reduced. The resource theory we focused on for DD was independent quantum instruments (IQI), where arbitrary temporally uncorrelated pre- and post-processing, plus temporal coarse-graining are allowed. Within IQI, DD can be understood as applying a resource transformation to a process tensor corresponding to a sequence of unitary rotations, followed by temporal coarse graining, e.g. to a regular two-time quantum channel. If DD is successful, the resultant channel should be less noisy than if no DD took place before the coarse-graining. 

\section{Problem with Choi Divergence Resource Quantifiers}
In Ref.~\cite{extractingdynamicalquantumresources} we mainly chose to measure the efficacy of DD in terms of the relative entropy between the Choi state of a process tensor and its own marginals -- which we shall refer to as Choi divergences. This led to quantifiers of total correlations, Markovian correlations, and non-Markovianity, $I,M,N$, respectively. A resource quantifier is only meaningful if it is a monotone under the free transformations of the theory. For DD in particular, this would correspond to sequences of unitary pulses plus temporal coarse-graining being unable to increase the values of $I,M,N$. We came to the conclusion that this was the case, stated as Thm.~1 in Ref.~\cite{extractingdynamicalquantumresources}, which was proved using Methods Prop.~1 (monotonicity under free superprocesses) and Supplementary Methods Prop.~1 (monotonicity under temporal coarse-graining). However, this relied on faulty implicit assumptions about the trace preservation of our resource transformations (they are trace preserving on process tensors, but not on general quantum states of the same dimensions). In Ref.~\cite{zambon2024processtensordistinguishabilitymeasures} it was shown that $I,M,N$ are not monotones in IQI. Thus, all conclusions that were drawn from the behaviour of $I,M,N$ must be re-examined.

\section{Purpose of this Letter}
In this letter, we provide a viable set of alternative monotones, from which all of our original conclusions about DD can still be drawn. We divide this letter into three parts. Firstly, we show that the notion of irreversibility examined in Ref.~\cite{extractingdynamicalquantumresources} Obs.~(1) allows one to define monotones $\overline{I}_{\hat{m}},\overline{M}_{\hat{m}},\overline{N}_{\hat{m}}$ that are suitable to replace $I,M,N$. We also prove that these new monotones satisfy a subadditivity relation analogous to the strict additivity relation $I=M+N$ (Eq.~(8) in Ref.~\cite{extractingdynamicalquantumresources}). Secondly, we show a correspondence between these monotones, and ones based on generalised comb divergences from Ref.~\cite{zambon2024processtensordistinguishabilitymeasures}. In particular, we show that $\overline{I}_{\hat{m}},\overline{M}_{\hat{m}},\overline{N}_{\hat{m}}$ are members of a family of generalised comb divergences, which in turn provides an operational interpretation to those generalised comb divergences. Finally, we re-examine the results of our numerical work in terms of $\overline{I}_{\hat{m}},\overline{M}_{\hat{m}},\overline{N}_{\hat{m}}$. This analysis allows us to confirm our original conclusions using valid monotones: that DD can be understood as temporal resource distillation using our new monotones, and that improvements to noise reduction via our MODD technique coincide with greater expenditure of the corresponding non-Markovianity monotone.

\section{Irreversibility Monotones}
The resource objects in the resources theories of Ref.~\cite{extractingdynamicalquantumresources} are process tensors~\cite{nonmarkovianquantumprocesses} $\mathbf{T}_{\hat{n}}$, corresponding to non-Markovian quantum processes accessed at some set of times $\hat{n}=\{ t_1,\dots,t_n \}$. These dynamical objects are represented as states using the Choi isomorphism~\cite{quantumstochasticprocesses}. The free transformations in these theories are superprocesses, which are a generalisation of superchannels, and map process tensors to other process tensors. In our notation, a superprocess $\mathbf{Z}_{\hat{n}\hat{m}}$ changes the set of intermediate times from $\hat{n}$ to $\hat{m} \subseteq \hat{n}$. This can be explicitly realised via temporal coarse-graining $ \mathbf{I}_{{\hat{n}}\setminus{\hat{m}}}$, which simply corresponds to `plugging' times in $\hat{n}$ but not in $\hat{m}$ with identity channels. For IQI in particular, the superprocesses correspond to temporally uncorrelated instruments. Thus, monotones in IQI are functions from the set of potential resources to non-negative real numbers that decrease under these transformations. 

We define the Choi divergences $I,M,N$ as in Ref.~\cite{extractingdynamicalquantumresources} Eq.~(7). Consider two distinct marginal processes of a given process $\mathbf{T}_{\hat{n}}$,
\begin{gather} \label{eq:marginals}
     \mathbf{T}^{\text{Mkv}}_{\hat{n}} \! := \! \bigotimes_{j=1}^{n+1}  {\rm tr}_{\bar{j}} \{ \mathbf{T}_{\hat{n}}\}
     \quad \mbox{and} \quad 
     \mathbf{T}^{\text{marg}}_{\hat{n}} \! := \!\!\! \bigotimes_{k=1}^{2(n+1)} \!\! {\rm tr}_{\bar{k}} \{ \mathbf{T}_{\hat{n}}\}.
\end{gather}
The index $j$ enumerates the constituent channels of the process, and $k$ splits this further into the marginal inputs and outputs of those channels. The overline on an index signifies its complement. As such, $\mathbf{T}^{\text{Mkv}}_{\hat{n}}$ only contains the Markovian temporal correlations of $\mathbf{T}_{\hat{n}}$, while  $\mathbf{T}^{\text{marg}}_{\hat{n}}$ contains no temporal correlations whatsoever. These are the nearest Markovian and uncorrelated processes to $\mathbf{T}_{\hat{n}}$, respectively, as measured by quantum relative entropy $S(x\|y) := \mbox{tr}\{x \log(x) - x \log(y)\}$. Given this, the total information $I$, Markovian information $M$, and non-Markovianity $N$ are
\begin{gather}
\begin{split}
    \label{eq:Imonotoneresolution}
    &I(\mathbf{T}_{\hat{n}}) \! := \! S\left(\mathbf{T}_{\hat{n}} \| \mathbf{T}^{\text{marg}}_{\hat{n}} \right), \ \
    M(\mathbf{T}_{\hat{n}}):=S\left( \mathbf{T}^{\text{Mkv}}_{\hat{n}} \| \mathbf{T}^{\text{marg}}_{\hat{n}} \right), \ \
    \mbox{and} \ \ N(\mathbf{T}_{\hat{n}}) \! := \! S\left(\mathbf{T}_{\hat{n}} \| \mathbf{T}^{\text{Mkv}}_{\hat{n}} \right)  .
\end{split}
\end{gather}
In Thm.~1 of Ref.~\cite{extractingdynamicalquantumresources}, based on the incorrect assumption that superprocesses and coarse graining correspond to trace-non-increasing maps, we claimed that $I, M$, and $N$ behaved monotonically under these operations. However, as shown in Ref.~\cite{zambon2024processtensordistinguishabilitymeasures}, this is not the case. 
This can be remedied by invoking the fact that the optimal value of any function -- such as $I,M,N$ -- that can be obtained subject to any allowed transformation must be non-increasing under temporal coarse-graining. Thus, we can replace $I, M$, and $N$ with three families of monotones.
\begin{theorem}[Irreversibility Monotones]
\label{thm::irrev}
Let $\hat{m} \subseteq \hat{n}$. Then,
\begin{gather}
\begin{split}
    &\overline{I}_{\hat{m}}(\mathbf{T}_{\hat{n}}) \! := \! \sup_{\mathbf{Z}_{\hat{n}\hat{n}} \in \mathsf{Z}_{\hat{n}\hat{n}}} I\Big( \llbracket \mathbf{T}_{\hat{n}} | \mathbf{Z}_{\hat{n}\hat{n}} | \mathbf{I}_{{\hat{n}}\setminus{\hat{m}}} \rrbracket  \Big), \ \
    \overline{M}_{\hat{m}}(\mathbf{T}_{\hat{n}}) \! := \! \sup_{\mathbf{Z}_{\hat{n}\hat{n}} \in \mathsf{Z}_{\hat{n}\hat{n}}} M\Big( \llbracket \mathbf{T}_{\hat{n}} | \mathbf{Z}_{\hat{n}\hat{n}} | \mathbf{I}_{{\hat{n}}\setminus{\hat{m}}} \rrbracket  \Big), \\
    &\mbox{\text{\emph{and}}} \ \ \overline{N}_{\hat{m}}(\mathbf{T}_{\hat{n}}) \! := \! \sup_{\mathbf{Z}_{\hat{n}\hat{n}} \in \mathsf{Z}_{\hat{n}\hat{n}}} N\Big( \llbracket \mathbf{T}_{\hat{n}} | \mathbf{Z}_{\hat{n}\hat{n}} | \mathbf{I}_{{\hat{n}}\setminus{\hat{m}}} \rrbracket  \Big),
\end{split}
\end{gather}
are monotones under free transformations in $\mathsf{Z}_{\hat{n}\hat{n}}$, where $\mathbf{I}_{{\hat{n}}\setminus{\hat{m}}}$ is the temporal coarse-graining transformation from times $\hat{n}$ to times $\hat{m}$.
\end{theorem}
Thm.~\ref{thm::irrev} holds true for \textit{any} set of free operations $\mathsf{Z}_{\hat{n}\hat{n}}$, in particular for the case where $\mathsf{Z}_{\hat{n}\hat{n}}$ is the set of uncorrelated sequences of channels and temporal coarse graining, which corresponds to IQI, the resource theory of temporal resolution considered in Ref.~\cite{extractingdynamicalquantumresources}.
\begin{proof}
The composition of the optimal superprocess $\mathbf{Z}^{*}_{\hat{n}\hat{n}}$ with some other arbitrary one $\mathbf{Z}_{\hat{n}\hat{n}}$ will never be more optimal than $\mathbf{Z}^{*}_{\hat{n}\hat{n}}$ alone, implying that these are monotonic under superprocesses that do not change the number of steps. Corollary 1 of Ref.~\cite{extractingdynamicalquantumresources} implies that $\overline{I}_{\hat{m}},\overline{M}_{\hat{m}},\overline{N}_{\hat{m}}$ are all monotonic under temporal coarse-graining. Hence, $\overline{I}_{\hat{m}},\overline{M}_{\hat{m}},\overline{N}_{\hat{m}}$ are all monotones in IQI.
\end{proof}
These monotones have the physical interpretation of being the highest value one can obtain for $I,M,N$ at temporal resolution $\hat{m}$, given process $\mathbf{T}_{\hat{n}}$ with temporal resolution $\hat{n}$. For example, when $\hat{m}=\emptyset$, i.e. coarse-graining to a channel, $\overline{I}_{\hat{m}}(\mathbf{T}_{\hat{n}})$ is the highest one can make the mutual information of that channel under any DD sequence. While $\overline{I}_{\hat{m}},\overline{M}_{\hat{m}},\overline{N}_{\hat{m}}$ are not as straightforward to compute as $I,M,N$, our multitime optimal dynamical decoupling (MODD) algorithm intrinsically approximates $\overline{I}_{\emptyset}$, and could be adapted to approximate $\overline{I}_{\hat{m}}$, $\overline{M}_{\hat{m}}$, or $\overline{N}_{\hat{m}}$ more generally.

In Eq.~(8) of Ref.~\cite{extractingdynamicalquantumresources}, we related $I,M,N$ by
\begin{equation}
    I=M+N,
\end{equation}
which implies that Markovian and non-Markovian contributions to the total temporal correlations necessarily come at the cost of each other. The same statement is not true in general for $\overline{I}_{\hat{m}},\overline{M}_{\hat{m}},\overline{N}_{\hat{m}}$. However, the superprocess $\mathbf{Z}^{M}_{\hat{n}\hat{n}}$ optimising $\overline{M}_{\hat{m}}$ will in general not be the same as the superprocess $\mathbf{Z}^{N}_{\hat{n}\hat{n}}$ optimising $\overline{N}_{\hat{m}}$. This leads to our next theorem.
\begin{theorem}[Subadditivity of $\overline{M}$ and $\overline{N}$] \label{thm:subadditivity}
For any process $\mathbf{T}_{\hat{n}}$, the monotones $\overline{I}_{\hat{m}},\overline{M}_{\hat{m}},\overline{N}_{\hat{m}}$ satisfy
    \begin{gather}
\overline{I}_{\hat{m}}(\mathbf{T}_{\hat{n}}) \leq    \overline{M}_{\hat{m}}(\mathbf{T}_{\hat{n}}) + \overline{N}_{\hat{m}}(\mathbf{T}_{\hat{n}}).     
\end{gather}
\end{theorem}
\begin{proof}
Consider the difference
 \begin{align}
     &\overline{I}_{\hat{m}}(\mathbf{T}_{\hat{n}}) - \overline{M}_{\hat{m}}(\mathbf{T}_{\hat{n}}) \\
     &= \sup_{\mathbf{Z}_{\hat{n}\hat{n}} \in \mathsf{Z}_{\hat{n}\hat{n}}} I\Big( \llbracket \mathbf{T}_{\hat{n}} | \mathbf{Z}_{\hat{n}\hat{n}} | \mathbf{I}_{{\hat{n} \setminus \hat{m}}} \rrbracket \Big) - \sup_{\mathbf{Z}'_{\hat{n}\hat{n}} \in \mathsf{Z}_{\hat{n}\hat{n}}} M\Big( \llbracket \mathbf{T}_{\hat{n}} | \mathbf{Z}'_{\hat{n}\hat{n}}  | \mathbf{I}_{{\hat{n} \setminus \hat{m}}} \rrbracket\Big) \\
     \label{eqn::subadditivity}
     &= \sup_{\mathbf{Z}_{\hat{n}\hat{n}} \in \mathsf{Z}_{\hat{n}\hat{n}}} S\left(\llbracket \mathbf{T}_{\hat{n}} | \mathbf{Z}_{\hat{n}\hat{n}} | \mathbf{I}_{{\hat{n} \setminus \hat{m}}} \rrbracket \| \llbracket \mathbf{T}_{\hat{n}} | \mathbf{Z}_{\hat{n}\hat{n}} | \mathbf{I}_{{\hat{n} \setminus \hat{m}}} \rrbracket^{\text{marg}}\right)  - 
     \sup_{\mathbf{Z}'_{\hat{n}\hat{n}} \in \mathsf{Z}_{\hat{n}\hat{n}}} S\left(\llbracket \mathbf{T}_{\hat{n}}| \mathbf{Z}'_{\hat{n}\hat{n}} | \mathbf{I}_{{\hat{n} \setminus \hat{m}}} \rrbracket^{\text{Mkv}} \|  \llbracket \mathbf{T}_{\hat{n}} | \mathbf{Z}'_{\hat{n}\hat{n}} | \mathbf{I}_{{\hat{n} \setminus \hat{m}}} \rrbracket^{\text{marg}} \right).
 \end{align}
 Assuming that $\mathsf{Z}_{\hat{n}\hat{n}}$ is compact (which holds for all sets $\mathsf{Z}_{\hat{n}\hat{n}}$ we consider), we denote by $\mathbf{Z}_{\hat{n}\hat{n}}^*$ the superprocess that maximises the first term in Eq.~\eqref{eqn::subadditivity}. Then we obtain
 \begin{align}
     &\overline{I}_{\hat{m}}(\mathbf{T}_{\hat{n}}) - \overline{M}_{\hat{m}}(\mathbf{T}_{\hat{n}}) \\
     \leq& S\left(\llbracket \mathbf{T}_{\hat{n}} | \mathbf{Z}^*_{\hat{n}\hat{n}} | \mathbf{I}_{{\hat{n} \setminus \hat{m}}} \rrbracket \| \llbracket \mathbf{T}_{\hat{n}} | \mathbf{Z}^*_{\hat{n}\hat{n}} | \mathbf{I}_{{\hat{n} \setminus \hat{m}}} \rrbracket^{\text{marg}}\right)  - 
     S\left(\llbracket \mathbf{T}_{\hat{n}}| \mathbf{Z}^*_{\hat{n}\hat{n}} | \mathbf{I}_{{\hat{n} \setminus \hat{m}}} \rrbracket^{\text{Mkv}} \|  \llbracket \mathbf{T}_{\hat{n}} | \mathbf{Z}^*_{\hat{n}\hat{n}} | \mathbf{I}_{{\hat{n} \setminus \hat{m}}} \rrbracket^{\text{marg}} \right) \\
     =& -S\left(\llbracket \mathbf{T}_{\hat{n}} | \mathbf{Z}^*_{\hat{n}\hat{n}} | \mathbf{I}_{{\hat{n} \setminus \hat{m}}} \rrbracket\right) + S\left(\llbracket \mathbf{T}_{\hat{n}} | \mathbf{Z}^*_{\hat{n}\hat{n}} | \mathbf{I}_{{\hat{n} \setminus \hat{m}}} \rrbracket^{\text{marg}} \right) + S\left(\llbracket \mathbf{T}_{\hat{n}} | \mathbf{Z}^*_{\hat{n}\hat{n}} | \mathbf{I}_{{\hat{n} \setminus \hat{m}}} \rrbracket^{\text{Mkv}} \right) - S\left(\llbracket \mathbf{T}_{\hat{n}} | \mathbf{Z}^*_{\hat{n}\hat{n}} | \mathbf{I}_{{\hat{n} \setminus \hat{m}}} \rrbracket^{\text{marg}} \right) \\
     =& S\left(\llbracket \mathbf{T}_{\hat{n}} | \mathbf{Z}^*_{\hat{n}\hat{n}} | \mathbf{I}_{{\hat{n} \setminus \hat{m}}} \rrbracket^{\text{Mkv}} \right) -S\left(\llbracket \mathbf{T}_{\hat{n}} | \mathbf{Z}^*_{\hat{n}\hat{n}} | \mathbf{I}_{{\hat{n} \setminus \hat{m}}} \rrbracket \right)\\
     \leq & \sup_{\mathbf{Z}_{\hat{n}\hat{n}} \in \mathsf{Z}_{\hat{n}\hat{n}}} \left[S\left(\llbracket \mathbf{T}_{\hat{n}} | \mathbf{Z}_{\hat{n}\hat{n}} | \mathbf{I}_{{\hat{n} \setminus \hat{m}}} \rrbracket^{\text{Mkv}} \right) -S\left(\llbracket \mathbf{T}_{\hat{n}} | \mathbf{Z}_{\hat{n}\hat{n}} | \mathbf{I}_{{\hat{n} \setminus \hat{m}}} \rrbracket \right) \right] \\
     =& \sup_{\mathbf{Z}_{\hat{n}\hat{n}} \in \mathsf{Z}_{\hat{n}\hat{n}}} \left[S\left(\ \llbracket \mathbf{T}_{\hat{n}} | \mathbf{Z}_{\hat{n}\hat{n}} | \mathbf{I}_{{\hat{n} \setminus \hat{m}}} \rrbracket  \| \llbracket \mathbf{T}_{\hat{n}} | \mathbf{Z}_{\hat{n}\hat{n}} | \mathbf{I}_{{\hat{n} \setminus \hat{m}}} \rrbracket^{\text{Mkv}} \right)\right]\\
     =& \overline{N}_{\hat{m}}(\mathbf{T}_{\hat{n}}).
 \end{align}
The third line comes from explicitly writing out the relative entropy between a quantum state and its own marginals. The second inequality holds because the value any arbitrary $\mathbf{Z}_{\hat{n}\hat{n}} \in \mathsf{Z}_{\hat{n}\hat{n}}$ cannot be greater than the supremum over $ \mathsf{Z}_{\hat{n}\hat{n}}$.
\end{proof}

In Thms.~{1-4} in the Supplementary Discussion of Ref.~\cite{extractingdynamicalquantumresources} the sequential/parallel composition of free resources is studied in terms of the trace distance and relative entropy between Choi states. While these supplementary results remain meaningful in their own right for understanding changes in the distinguishability of multitime processes under composition with free processes, they cannot be understood as pertaining to the behaviour of the resource monotones $\overline{I}_{\hat{m}},\overline{M}_{\hat{m}},$ and $\overline{N}_{\hat{m}}$.

Here, we generalise these results to the monotone $\overline{I}$, i.e., we show that $\overline{I}$ remains invariant under sequential and parallel composition with free processes. Similar arguments can be applied to $\overline{M}$ and $\overline{N}$. The trade-off between total correlation, non-Markovianity, and Markovian correlation is rather nontrivial, see Ref.~\cite{Zambon_2024} for a detailed discussion.

First, we consider sequential composition of a process $\mathbf{T}_{\hat{n}}$ on $n$ times with a free process $\mathbf{S}_{\hat{n}'}$ on $n'$ times. We emphasise that this parallel composition leads to a `new' time slot in-between the two processes, which we denote by $t'$, such that the resulting monotone is of the form $\overline{I}_{(\hat{m}, t' , \hat{m}')}$.
\begin{proposition}[Sequential Composition]
    $\overline{I}_{\hat{m}}$ is invariant under parallel composition with free processes in IQI.
\end{proposition}
\begin{proof}
Let $\mathbf{T}_{\hat{n}}$ be an arbitrary process, and $\mathbf{S}_{\hat{n}'}$ a free one. Then,\begin{align}
    \overline{I}_{(\hat{m}, t' , \hat{m}')}( \mathbf{S}_{\hat{n}'} \circ \mathbf{T}_{\hat{n}} ) =&  \sup_{\mathbf{Z}_{(\hat{n} , \hat{n}')(\hat{n} , \hat{n}')} } I\Big( \llbracket \mathbf{S}_{\hat{n}'} \circ \mathbf{T}_{\hat{n}} | \mathbf{Z}_{(\hat{n}'\circ \hat{n}) (\hat{n} , \hat{n}')} | \mathbf{I}_{{(\hat{n} , \hat{n}')}\setminus{(\hat{m}, t' , \hat{m}')}} \rrbracket  \Big) \\
    =& \sup_{\mathbf{Z}_{\hat{n}'\hat{n}'} \circ \mathbf{Z}_{\hat{n}\hat{n}}  } I\Big( \llbracket \mathbf{S}_{\hat{n}'} \circ \mathbf{T}_{\hat{n}} | \mathbf{Z}_{\hat{n}'\hat{n}'} \circ \mathbf{Z}_{\hat{n}\hat{n}}  | \mathbf{I}_{{(\hat{n} , \hat{n}')}\setminus{(\hat{m}, t' , \hat{m}')}} \rrbracket  \Big) \\
    =& \sup_{\mathbf{Z}_{\hat{n}'\hat{n}'} , \mathbf{Z}_{\hat{n}\hat{n}}  } I\Big( \llbracket \mathbf{S}_{\hat{n}'}  | \mathbf{Z}_{\hat{n}'\hat{n}'}  | \mathbf{I}_{{ \hat{n}'}\setminus{ \hat{m}'}} \rrbracket \otimes \llbracket \mathbf{T}_{\hat{n}} |  \mathbf{Z}_{\hat{n}\hat{n}}  | \mathbf{I}_{{\hat{n}}\setminus{\hat{m}}} \rrbracket  \Big) \\
    =& \sup_{\mathbf{Z}_{\hat{n}'\hat{n}'} } I\Big( \llbracket \mathbf{S}_{\hat{n}\prime} | \mathbf{Z}_{\hat{n}'\hat{n}'} | \mathbf{I}_{ \hat{n}' \setminus \hat{m}'} \rrbracket  \Big) + \sup_{\mathbf{Z}_{\hat{n}\hat{n}} } I\Big( \llbracket \mathbf{T}_{\hat{n}} | \mathbf{Z}_{\hat{n}\hat{n}} | \mathbf{I}_{{\hat{n}}\setminus\hat{m}} \rrbracket  \Big) \\
    =& \sup_{\mathbf{Z}_{\hat{n}\hat{n}} } I\Big( \llbracket \mathbf{T}_{\hat{n}} | \mathbf{Z}_{\hat{n}\hat{n}} | \mathbf{I}_{{\hat{n}}\setminus\hat{m}} \rrbracket  \Big) \\
    =& \overline{I}_{\hat{m}}( \mathbf{T}_{\hat{n}} ).
\end{align}
The second equality is due to the fact that superprocesses in IQI are entirely uncorrelated between different times. The fourth equality holds due to additivity of entropy for uncorrelated subsystems, and because sequential composition does not impart correlations. The penultimate line follows from the fact that $\mathbf{S}_{\hat{n}'}$ is a free process, and hence has zero monotone value according to $\overline{I}$. 
\end{proof}

We move on to parallel composition. 
\begin{proposition}[Parallel Composition]
    $\overline{I}_{\hat{m}}$ is invariant under parallel composition with free processes in IQI.
\end{proposition}
\begin{proof}
We begin by showing that $\overline{I}_{\hat{m}}( \mathbf{T}_{\hat{n}} \otimes \mathbf{S}_{\hat{n}} )  \geq \overline{I}_{\hat{m}}( \mathbf{T}_{\hat{n}})$, i.e., $\overline{I}_{\hat{m}}$ is is non-decreasing under parallel composition with free processes $\mathbf{S}_{\hat{n}}$. We label the subsystems of original process $\mathbf{T}_{\hat{n}}$ and free process $\mathbf{S}_{\hat{n}}$ as $A$ and $B$ respectively. With this, we have
\begin{align}
    \overline{I}_{\hat{m}}( \mathbf{T}_{\hat{n}} \otimes \mathbf{S}_{\hat{n}} ) 
    =& \sup_{\mathbf{Z}^{AB}_{\hat{n}\hat{n}} } I\Big( \llbracket \mathbf{T}^{A}_{\hat{n}} \otimes \mathbf{S}^{B}_{\hat{n}} | \mathbf{Z}^{AB}_{\hat{n}\hat{n}}| \mathbf{I}^{AB}_{{\hat{n}}\setminus{\hat{m}}} \rrbracket  \Big) \\
    \geq & \sup_{\mathbf{Z}^{A}_{\hat{n}\hat{n}},\mathbf{Z}'^{B}_{\hat{n}\hat{n}} } I\Big( \llbracket \mathbf{T}^{A}_{\hat{n}} \otimes \mathbf{S}^{B}_{\hat{n}} | \mathbf{Z}^{A}_{\hat{n}\hat{n}} \otimes \mathbf{Z}'^{B}_{\hat{n}\hat{n}}| \mathbf{I}^{AB}_{{\hat{n}}\setminus{\hat{m}}} \rrbracket  \Big) \\
    =& \sup_{\mathbf{Z}^{A}_{\hat{n}\hat{n}} } I\Big( \llbracket \mathbf{T}^{A}_{\hat{n}}  | \mathbf{Z}^{A}_{\hat{n}\hat{n}}| \mathbf{I}^{A}_{{\hat{n}}\setminus{\hat{m}}} \rrbracket  \Big) + \sup_{\mathbf{Z}'^{B}_{\hat{n}\hat{n}} } I\Big( \llbracket \mathbf{S}^{B}_{\hat{n}}  | \mathbf{Z}'^{B}_{\hat{n}\hat{n}}| \mathbf{I}^{B}_{{\hat{n}}\setminus{\hat{m}}} \rrbracket  \Big) \\
    =& \sup_{\mathbf{Z}^{A}_{\hat{n}\hat{n}} } I\Big( \llbracket \mathbf{T}^{A}_{\hat{n}}  | \mathbf{Z}^{A}_{\hat{n}\hat{n}}| \mathbf{I}^{A}_{{\hat{n}}\setminus{\hat{m}}} \rrbracket  \Big)  \\
    =&\overline{I}_{\hat{m}}( \mathbf{T}_{\hat{n}}).
\end{align}
The inequality holds because the supremum over product superprocesses $\mathbf{Z}^{A}_{\hat{n}\hat{n}} \otimes \mathbf{Z}'^{B}_{\hat{n}\hat{n}}$ (instead of over \textit{all} superprocesses $\mathbf{Z}^{AB}_{\hat{n}\hat{n}}$ that IQI permits) is not necessarily optimal. The second equality holds due to the additivity of entropy for uncorrelated subsystems. 

Now we show that monotone value is non-increasing. To do so, we employ the fact that free processes $\mathbf{S}^{B}_{\hat{n}}$ in IQI are completely uncorrelated in time, i.e, they are of the form $(\mathbbm{1} \otimes \rho_1)\otimes \dots \otimes (\mathbbm{1} \otimes \rho_{n+1})$. Consequently, $ \llbracket \mathbf{S}^{B}_{\hat{n}} | \mathbf{Z}^{AB}_{\hat{n}\hat{n}}  \rrbracket = \mathbf{Z}^{\prime A}_{\hat{n}\hat{n}}$ is a free superprocess in IQI and we have $ \llbracket \mathbf{S}^{B}_{\hat{n}} | \mathsf{Z}^{AB}_{\hat{n}\hat{n}} \rrbracket \subseteq \mathsf{Z}^{A}_{\hat{n}\hat{n}}$. Here, $B$ is subsumed into a larger effective $A$ subsystem, and $ \mathbf{Z}^{\prime A}_{\hat{n}\hat{n}}$ is still unable to transmit correlations on this effective $A$ subsystem through time. With this, we can show:
\begin{align}
    \overline{I}_{\hat{m}}( \mathbf{T}_{\hat{n}} \otimes \mathbf{S}_{\hat{n}} ) =& \sup_{\mathbf{Z}^{AB}_{\hat{n}\hat{n}} } I\Big( \llbracket \mathbf{T}^{A}_{\hat{n}} \otimes \mathbf{S}^{B}_{\hat{n}} | \mathbf{Z}^{AB}_{\hat{n}\hat{n}}| \mathbf{I}^{AB}_{{\hat{n}}\setminus{\hat{m}}} \rrbracket  \Big) \\
    =&  \sup_{\mathbf{Z}^{\prime A}_{\hat{n}\hat{n}} } I\Big( \llbracket \mathbf{T}^{A}_{\hat{n}} | \mathbf{Z}^{\prime A}_{\hat{n}\hat{n}}| \mathbf{I}^{A}_{{\hat{n}}\setminus{\hat{m}}} \rrbracket  \Big)  \\
    \leq & \sup_{\mathbf{Z}^{A}_{\hat{n}\hat{n}} \in \mathsf{Z}_{\hat{n}\hat{n}}^{A}} I\Big( \llbracket \mathbf{T}^{A}_{\hat{n}}  | \mathbf{Z}^{A}_{\hat{n}\hat{n}}| \mathbf{I}^{A}_{{\hat{n}}\setminus{\hat{m}}} \rrbracket  \Big)  \\
    =&\overline{I}_{\hat{m}}( \mathbf{T}_{\hat{n}}),
\end{align}
where we have used in the second equality that free processes in IQI do not contain any correlations in time, and thus do not lead to superprocesses $\mathbf{Z}_{\hat{n}\hat{n}}^{\prime A}$ that lie outside of what is possible within IQI. 
\end{proof}

\section{Correspondence with Generalised Comb Divergences}
The irreversibility monotones $\overline{I}_{\hat{m}}(\mathbf{T}_{\hat{n}}),     \overline{N}_{\hat{m}}(\mathbf{T}_{\hat{n}})$ and $\overline{M}_{\hat{m}}(\mathbf{T}_{\hat{n}})$ we introduced in Thm.~\ref{thm::irrev} above decrease monotonically under the free superprocesses of IQI and temporal coarse-graining. Analogously, the \textit{generalised comb divergences} introduced in Ref.~\cite{zambon2024processtensordistinguishabilitymeasures} decrease monotonically under \textit{all} superprocesses $\mathbf{Z}_{\hat{n}\hat{m}}$. Instead of optimising over superprocesses as we do in this manuscript, these generalised comb divergences optimise over an input comb. Adapting generalised comb divergences to our setting, we define
\begin{equation}
\label{eqn::gen_div}
    D\big(  \mathbf{T}_{\hat{n}} \big\Vert \mathbf{R}_{\hat{n}} \big) := \sup_{\mathbf{S}_{\hat{n}} \in \mathsf{S}_{\hat{n}}} S\Big( \llbracket \mathbf{T}_{\hat{n}}| \mathbf{S}_{\hat{n}}\rrbracket \Big\Vert \llbracket \mathbf{R}_{\hat{n}}|\mathbf{S}_{\hat{n}}\rrbracket \Big),
\end{equation}
where $\mathbf{S}_{\hat{n}}$ is a control comb that is such that $\llbracket \mathbf{T}_{\hat{n}}| \mathbf{S}_{\hat{n}}\rrbracket$ and $\llbracket \mathbf{R}_{\hat{n}}||\mathbf{S}_{\hat{n}}\rrbracket$ correspond to quantum channels, in contrast to Ref.~\cite{zambon2024processtensordistinguishabilitymeasures}, where these correspond to quantum \textit{states}. This difference is due to the fact that our primary concern is the noisiness of this resultant channel. In either case, the $\mathbf{S}_{\hat{n}}$ comb may contain auxiliary subsystems that do not directly interact with $\mathbf{T}_{\hat{n}}$. 

Here, we show that the resource monotones introduced in Thm.~\ref{thm::irrev} are a special case of Eq.~\eqref{eqn::gen_div}, with the only difference that the supremum is taken over a limited set of control combs $\mathsf{S}^\textup{reach}_{\hat{n}}$ instead of arbitray control combs. To this end, we first note that, given sets of superprocesses $\mathsf{Z}_{\hat{n}\hat{m}}$, possibly changing the temporal resolution, we can obtain the corresponding sets $\mathsf{S}^\text{reach}_{\hat{n}}$ of \textit{reachable} control combs as those that can be obtained from $\mathsf{Z}_{\hat{n}\hat{n}}$ via coarse graining, i.e., 
\begin{gather}
    \mathsf{S}^\text{reach}_{\hat{n}} := \{ \mathbf{S}_{\hat{n}} = \llbracket \mathbf{Z}_{\hat{n}\hat{m}}| \mathbf{I}_{\hat{m}}\rrbracket \ \ \text{for some } \mathbf{Z}_{\hat{n}\hat{m}} \in \mathsf{Z}_{\hat{n}\hat{m}}\}.
\end{gather}
In what follows we only consider \textit{compatible} sets $\mathsf{Z}_{\hat{n}\hat{m}}$ of superprocesses that are connected via temporal coarse graining between different temporal resolutions, i.e., $\mathsf{Z}_{\hat{n}\hat{m}} = \llbracket \mathsf{Z}_{\hat{n}\hat{n}}|\mathbf{I}_{\hat{n}\setminus \hat{m}}\rrbracket$ for all $\hat{m}\subseteq \hat{n}$, and are closed under composition, i.e., $\mathsf{Z}_{\hat{n}\hat{m}} \mathsf{Z}_{\hat{m}\hat{\ell}} \subseteq \mathsf{Z}_{\hat{n}\hat{\ell}}$ for all $\hat{n}, \hat{m}$ and $\hat{\ell}$. With this, we have the following Lemma:

\begin{lemma}
Let $\mathsf{Z}_{\hat{n}\hat{m}}$ be compatible sets of superprocesses, and let $\mathsf{S}^\textup{reach}_{\hat{n}}$ be the corresponding sets of reachable control combs. Then the \textit{reachable} comb divergence 
\begin{gather}
    D^\textup{reach}\big( \mathbf{T}_{\hat{n}} \big\Vert \mathbf{R}_{\hat{n}} \big) := \sup_{\mathbf{S}_{\hat{n}} \in \mathsf{S}^\textup{reach}_{\hat{n}}} S\Big( \llbracket \mathbf{T}_{\hat{n}}| \mathbf{S}_{\hat{n}}\rrbracket \Big\Vert \llbracket \mathbf{R}_{\hat{n}}|\mathbf{S}_{\hat{n}}\rrbracket \Big),
\end{gather}
decreases monotonically under all superprocesses in $\mathsf{Z}_{\hat{n}\hat{m}}$.
\end{lemma}
For the particular case of IQI predominantly considered in Ref.~\cite{extractingdynamicalquantumresources}, any Markovian control comb is within $\mathsf{S}^{\text{reach}}_{\hat{n}}$.
\begin{proof}
    The proof follows the same lines as that of Thm. 1 in Ref.~\cite{zambon2024processtensordistinguishabilitymeasures}. Let $\mathbf{Z}_{\hat{n}\hat{m}} \in \mathsf{Z}_{\hat{n}\hat{m}}$. We have 
    \begin{align}
    D^\text{reach}\big( \llbracket \mathbf{T}_{\hat{n}}| \mathbf{Z}_{\hat{n}\hat{m}}\rrbracket \big\Vert \llbracket\mathbf{R}_{\hat{n}} |   \mathbf{Z}_{\hat{n}\hat{m}}\rrbracket\big) &= \sup_{\mathbf{S}_{\hat{m}} \in \mathsf{S}^\text{reach}_{\hat{m}}} S\Big( \llbracket \mathbf{T}_{\hat{n}}|\mathbf{Z}_{\hat{n}\hat{m}}| \mathbf{S}_{\hat{m}}\rrbracket \Big\Vert \llbracket \mathbf{R}_{\hat{n}}|\mathbf{Z}_{\hat{n}\hat{m}}|\mathbf{S}_{\hat{m}}\rrbracket \Big) \\
    & = \sup_{\mathbf{Z}'_{\hat{m}\hat{m}} \in \mathsf{Z}_{\hat{m}\hat{m}}} S\Big( \llbracket \mathbf{T}_{\hat{n}}|\mathbf{Z}_{\hat{n}\hat{m}}\mathbf{Z}'_{\hat{m}\hat{m}}|\mathbf{I}_{\hat{m}}\rrbracket \Big\Vert \llbracket \mathbf{R}_{\hat{m}}|\mathbf{Z}_{\hat{n}\hat{m}}\mathbf{Z}'_{\hat{m}\hat{m}}|\mathbf{I}_{\hat{m}}\rrbracket \Big) \\
    \label{eqn::reachDiv}
    &\leq \sup_{\mathbf{Z}'_{\hat{n}\hat{m}} \in \mathsf{Z}_{\hat{n}\hat{m}}} S\Big( \llbracket \mathbf{T}_{\hat{n}}|\mathbf{Z}'_{\hat{n}\hat{m}}|\mathbf{I}_{\hat{m}}\rrbracket \Big\Vert \llbracket \mathbf{R}_{\hat{m}}|\mathbf{Z}'_{\hat{n}\hat{m}}|\mathbf{I}_{\hat{m}}\rrbracket \Big) \\
    \label{eqn::reachDiv2}
    &= \sup_{\mathbf{S}_{\hat{n}} \in \mathsf{S}^\text{reach}_{\hat{n}}} S\Big( \llbracket \mathbf{T}_{\hat{n}}|\mathbf{S}_{\hat{n}}\rrbracket \Big\Vert \llbracket \mathbf{R}_{\hat{m}}|\mathbf{S}_{\hat{n}}\rrbracket \Big) = D^\textup{reach}\big( \mathbf{T}_{\hat{n}} \big\Vert \mathbf{R}_{\hat{n}} \big),
    \end{align}
where, for the last two lines, we have used the fact that $\mathbf{Z}_{\hat{n}\hat{m}}\mathsf{Z}_{\hat{m}\hat{m}} \subseteq \mathsf{Z}_{\hat{n}\hat{m}}$, as well as that $\llbracket\mathsf{Z}_{\hat{n}\hat{m}}|\mathbf{I}_{\hat{m}}\rrbracket = \mathsf{S}_{\hat{n}}^\text{reach}$, which both follow from the from the compatibility of the sets $\mathsf{Z}_{\hat{n}\hat{m}}$.
\end{proof}
As is evident from the proof of this Lemma [in particular in the step from Eq.~\eqref{eqn::reachDiv} to Eq.~\eqref{eqn::reachDiv2}] optimization over sets of control combs, as required for the computation of reachable comb divergences, is equivalent to the optimization over superprocesses $\mathsf{Z}_{\hat{n}\hat{m}}$ and subsequent temporal coarse graining. Consequently, for $\hat{m} = \hat{n}$, the irreversibility monotones $\overline{I}_{\hat{m}}(\mathbf{T}_{\hat{n}}), \overline{N}_{\hat{m}}(\mathbf{T}_{\hat{n}})$ and $\overline{M}_{\hat{m}}(\mathbf{T}_{\hat{n}})$, which follow from an optimization over superprocesses $\mathbf{Z}_{\hat{n}\hat{n}}$ can be understood as following from a reachable comb divergence -- and thus a special version of the generalised comb divergences introduced in Ref.~\cite{zambon2024processtensordistinguishabilitymeasures}. Comparing generalised and reachable comb divergences, monotonicity is only enforced by the former under free transformations of the particular resource theory, as opposed to any arbitrary superprocess with the latter.

\section{Interpretation of Numerical Results of Ref.~\cite{extractingdynamicalquantumresources}}
In Ref.~\cite{extractingdynamicalquantumresources} we used $I,M,N$ to analyse DD in numerical experiments. From these, we concluded that non-Markovianity was consumed in DD, that DD can be understood in terms of temporal resource distillation, and that our multitimescale optimal dynamical decoupling (MODD) technique achieves a high efficacy in noise reduction. However, since $I,M,N$ have been shown to not be monotones in IQI, we must re-examine these empirical conclusions in terms of $\overline{I}_{\hat{m}},\overline{M}_{\hat{m}},\overline{N}_{\hat{m}}$ instead. In doing this, we confirm our original claims.

Multitimescale optimal dynamical decoupling (MODD)~\cite{extractingdynamicalquantumresources} is a method to numerically approximate the optimal superprocess $\mathbf{Z}^{I}_{\hat{n}\hat{n}}$ whose relative entropy is equal to the monotone value
\begin{equation}
   I_{\hat{m}} \Big( \llbracket \mathbf{T}_{\hat{n}} | \mathbf{Z}^{\text{MODD}}_{\hat{n}\hat{n}} | \mathbf{I}_{{\hat{n}}\setminus{\hat{m}}} \rrbracket \Big) \approx I_{\hat{m}} \Big( \llbracket \mathbf{T}_{\hat{n}} | \mathbf{Z}^{I}_{\hat{n}\hat{n}} | \mathbf{I}_{{\hat{n}}\setminus{\hat{m}}} \rrbracket \Big) =\overline{I}_{\hat{m}} (  \mathbf{T}_{\hat{n}} )  .
\end{equation}
In our simulations we optimise $\overline{I}_{\emptyset} (  \mathbf{T}_{\hat{n}} )$ in particular -- the mutual information of the resultant channel. Putting the topic of resource theories aside, this is a convenient indicator of the quality of a quantum channel. Ref.~\cite{extractingdynamicalquantumresources} Fig.~4b corresponds to an experiment with $|\hat{n}|=3$ intermediate times, coarse-graining to $|\emptyset|=0$ and measuring the mutual information of the resultant channel. Fig.~4c is the same but starting with $|\hat{n}|=15$ intermediate times. The top lines on Fig.~4b and Fig.~4c can be understood as numerical approximations of $\overline{I}_{\emptyset} (  \mathbf{T}_{\hat{n}} )$. In these graphs, we see that MODD saturates the potential value of the mutual information $\overline{I}_{\hat{m}} (  \mathbf{T}_{\emptyset} )\approx2$. While DD cannot increase monotone value, the optimal superprocess (here approximately obtained via MODD) can concentrate resource value amongst fewer temporal subsystems, corresponding to a temporal form of resource distillation. 

The other lines on Fig.~4b and Fig.~4c -- corresponding to standard DD, no intervention, concatenated DD -- are all below what is achieved by MODD (as well as ODD at a single timescale), especially for long times. This can be interpreted as loss occurring during the resource distillation process, i.e. that not all of the initial resource value is maintained.

The results of Fig.~4d can be interpreted as the finding that greater success at DD coincides with greater consumption of non-Markovianity in the initial resource. However, reaching this conclusions requires some careful thought. Due their monotonicity, $\overline{I}_{\hat{m}},\overline{M}_{\hat{m}},\overline{N}_{\hat{m}}$ can only decrease under the free transformations of IQI. However -- by definition -- there is (at least) one transformation that will not decrease the value of an irreversibility monotone. These are the superprocesses $\mathbf{Z}^{I}_{\hat{n}\hat{n}},\mathbf{Z}^{M}_{\hat{n}\hat{n}},\mathbf{Z}^{N}_{\hat{n}\hat{n}}$ that maximise each of their respective monotones. Critically, these superprocesses need not be the same for different monotones, and hence optimising $\overline{M}_{\hat{m}}$ in general requires sacrificing $\overline{N}_{\hat{m}}$. It is in this sense that non-Markovianity is `expended' during DD.

We observe this phenomenon in Fig.~4d. This plots the difference in the value of $I,M,N$ between the the coarse-grained process after doing nothing $\llbracket \mathbf{T}_{\hat{n}} | \mathbf{I}_{{\hat{n}}\setminus{\hat{m}}} \rrbracket$ and one of $\llbracket \mathbf{T}_{\hat{n}} | \mathbf{Z}^{\text{DD}}_{\hat{n}\hat{n}} | \mathbf{I}_{{\hat{n}}\setminus{\hat{m}}} \rrbracket$ or $\llbracket \mathbf{T}_{\hat{n}} | \mathbf{Z}^{\text{MODD}}_{\hat{n}\hat{n}} | \mathbf{I}_{{\hat{n}}\setminus{\hat{m}}} \rrbracket$, 
\begin{equation}
    \Delta I_{DD}=I\Big( \llbracket \mathbf{T}_{\hat{n}} | \mathbf{Z}^{\text{DD}}_{\hat{n}\hat{n}} | \mathbf{I}_{{\hat{n}}\setminus{\hat{m}}} \rrbracket \Big)-I\Big(\llbracket \mathbf{T}_{\hat{n}}  |\mathbf{I}_{{\hat{n}}\setminus{\hat{m}}} \rrbracket \Big), \quad \Delta I_{MODD}=I\Big( \llbracket \mathbf{T}_{\hat{n}} | \mathbf{Z}^{\text{MODD}}_{\hat{n}\hat{n}} | \mathbf{I}_{{\hat{n}}\setminus{\hat{m}}} \rrbracket \Big)-I\Big(\llbracket \mathbf{T}_{\hat{n}} | \mathbf{I}_{{\hat{n}}\setminus{\hat{m}}} \rrbracket \Big),
\end{equation}
and similar for $M,N$, where coarse-graining is from $|\hat{n}|=15$ intermediate times to $|\hat{m}|=3$. $ \mathbf{Z}^{\text{DD}}_{\hat{n}\hat{n}}$ corresponds to the prototypical $X,Z,X,Z$ DD pulse sequence, and $\mathbf{Z}^{\text{MODD}}_{\hat{n}\hat{n}}$ is MODD. Note that, since $I,M,N$ are not monotones $I\Big(\llbracket \mathbf{T}_{\hat{n}} |  \mathbf{I}_{{\hat{n}}\setminus{\hat{m}}} \rrbracket \Big),M\Big(\llbracket \mathbf{T}_{\hat{n}} | \mathbf{I}_{{\hat{n}}\setminus{\hat{m}}} \rrbracket \Big),N\Big(\llbracket \mathbf{T}_{\hat{n}} | \mathbf{I}_{{\hat{n}}\setminus{\hat{m}}} \rrbracket \Big)$ do not have any special meaning in IQI, and are not relevant to our analysis. As such, $\Delta I_{MODD}$ can be understood as an approximation of $\overline{I}$, subtracting an arbitrary offset of $I\Big(\llbracket \mathbf{T}_{\hat{n}}  |\mathbf{I}_{{\hat{n}}\setminus{\hat{m}}} \rrbracket \Big)$. By comparison $\Delta I_{DD}$ is a lower bound on $\overline{I}$, shifted by the same offset. The same can be said of $\overline{M}$ and $\overline{N}$. 

Importantly, what can be seen on Fig.~4d is that at long timescales MODD achieves higher values for $I$ and $M$ than DD does, and lower values for $N$. Since MODD is optimising $\overline{I}_{\emptyset} (  \mathbf{T}_{\hat{n}} )=\overline{M}_{\emptyset} (  \mathbf{T}_{\hat{n}} )$, we do not guarantee that it also optimises $\overline{I}_{\hat{m}}$ or $\overline{M}_{\hat{m}}$ for $|\hat{m}|=3$. However, at very least it provides a lower bound on $\overline{M}_{\hat{m}}$, and an approximation of the optimal superprocess $\mathbf{Z}^{M}_{\hat{n}\hat{n}}$. What we see then, is that our approximation of $\mathbf{Z}^{M}_{\hat{n}\hat{n}}$ yields a lower $N$ value than standard DD $\mathbf{Z}^{\text{DD}}_{\hat{n}\hat{n}}$. This suggests that, $\mathbf{Z}^{M}_{\hat{n}\hat{n}}$ causes the non-Markovianity monotone $\overline{N}_{\hat{m}}$ to decrease. This is empirical evidence for the phenomenon that optimising $\overline{M}_{\hat{m}}$ sacrifices $\overline{N}_{\hat{m}}$.

\section{Discussion}
In Ref.~\cite{extractingdynamicalquantumresources} we framed DD in a resource theoretic manner and analysed the trade-off of total, Markovian, and non-Markovian correlations -- quantified by $I$, $M$, and  $N$ respectively -- during the decoupling process. As shown in Ref.~\cite{zambon2024processtensordistinguishabilitymeasures}, the quantities $I$, $M$, and $N$ are not monotones under the resource transformations we considered, putting the quantitative and qualitative conclusions of Ref.~\cite{extractingdynamicalquantumresources} in question. Here, building on the results of Ref.~\cite{zambon2024processtensordistinguishabilitymeasures}
we showed that the mathematical issues can be overcome by replacing $I$, $M$, and $N$ with $\overline{I}$, $\overline{M}$, and $\overline{N}$, which are indeed monotones under free transformations -- and thus valid quanitifiers of the expended resources in DD scenarios -- and satisfy subadditivity (instead of additivity, as was the case for $I$, $M$, and $N$). Additionally, as we showed, they can be directly interpreted as special case of the generalized comb divergences introduced in Ref.~\cite{zambon2024processtensordistinguishabilitymeasures}. Crucially, these changes do not alter the empirical results of Ref.~\cite{extractingdynamicalquantumresources}. Firstly, the conclusion that the DD protocols we propose outperform traditional ones remains entirely unaffected. Secondly, the analysis of he expenditure of non-Markovian correlations $N$ for overall correlations $I$, while not carried out with respect to the correct monotones $\overline{N}$ and $\overline{I}$ still provides a good proxy for the relationship of these quantities over the course of DD and thus correctly displays their qualitative behaviour.

\section{Acknowledgements}
We thank Guilherme Zambon for bringing the error in Ref.~\cite{extractingdynamicalquantumresources} to our attention, and for fruitful discussions on this topic.
SM acknowledges funding from the European Union's Horizon Europe research and innovation programme under the Marie Sk{\l}odowska-Curie grant agreement No.\ 101068332.

\bibliographystyle{apsrev4-1_custom}
\bibliography{refs.bib}

\end{document}